\documentclass[11pt]{article}

\usepackage[margin=1in]{geometry}
\usepackage{moi}

\newcommand{\Ham}{\mathrm{Ham}}
\newcommand{\dist}{\mathrm{dist}}
\newcommand{\bits}{\{0,1\}}

\newcommand{\cX}{\mathcal{X}}
\newcommand{\cF}{\mathcal{F}}
\newcommand{\cP}{\mathcal{P}}
\newcommand{\cQ}{\mathcal{Q}}

\title{Quantum Sketches, Hashing, and Approximate Nearest Neighbors}
\author{Sajjad Hashemian}
\date{\email{sajjadhasehmian@ut.ac.ir}}

\begin{document}
\maketitle

\begin{abstract}
Motivated by Johnson--Lindenstrauss dimension reduction, amplitude encoding, and the view of measurements as hash-like primitives, one might hope to compress an $n$-point approximate nearest neighbor (ANN) data structure into $O(\log n)$ qubits.
We rule out this possibility in a broad quantum sketch model, the dataset $P$ is encoded as an $m$-qubit state $\rho_P$, and each query is answered by an arbitrary query-dependent measurement on a fresh copy of $\rho_P$.
For every approximation factor $c\ge 1$ and constant success probability $p>1/2$, we exhibit $n$-point instances in Hamming space $\{0,1\}^d$ with $d=\Theta(\log n)$ for which any such sketch requires $m=\Omega(n)$ qubits, via a reduction to quantum random access codes and Nayak's lower bound.
These memory lower bounds coexist with potential quantum query-time gains and in candidate-scanning abstractions of hashing-based ANN, amplitude amplification yields a quadratic reduction in candidate checks, which is essentially optimal by Grover/BBBV-type bounds.
\end{abstract}

\section{Introduction}

Nearest neighbor search is a simple question with a large practical footprint.
One is given a dataset $P$ of points, later receives a query point $q$, and wants to return a nearby dataset point.
In high dimensions, exact nearest neighbor search is notoriously expensive.
Approximate nearest neighbor (ANN) relaxes correctness, for an approximation factor $c\ge 1$, it suffices to return a point within factor $c$ of the optimum distance.
This relaxation is strong enough to support many algorithmic ideas and weak enough to remain useful in applications.
A foundational paradigm is locality-sensitive hashing (LSH), introduced for ANN by Indyk and Motwani \cite{IM98} and developed extensively thereafter.
Random-hyperplane hashing for cosine similarity is a canonical example analyzed by Charikar \cite{Cha02}.

Quantum information suggests a different style of summary.
A state on $m$ qubits lives in a complex vector space of dimension $2^m$, which makes it natural to ask whether an entire dataset can be stored as one short quantum state.
The Johnson--Lindenstrauss lemma \cite{JL84} reinforces this intuition: after a JL reduction to $\Theta(\log n)$ dimensions, amplitude encoding can represent a single point using only $O(\log\log n)$ qubits, and one might imagine that query-dependent measurements act like hash functions.
The strongest form of this hope is that a worst-case ANN data structure for $n$ points can be compressed into $O(\log n)$ qubits.

This paper shows that the above compression dream fails in a very broad model.
The obstruction is not geometric.
It is informational.
We build a family of datasets indexed by bit strings $x\in\bits^n$ and a corresponding set of queries so that the correct approximate nearest neighbor answer to query $i$ reveals bit $x_i$.
Any sketch that answers these queries therefore gives a quantum random access code (QRAC) for $n$ bits, and Nayak's lower bound \cite{Nay99} forces $\Omega(n)$ qubits.

At the same time, our lower bound is not an argument against quantum speedups for similarity search in general.
If the dataset is stored in classical memory (or accessed via coherent oracles) and hashing produces a candidate set of size $M$, then Grover-type search can reduce $M$ candidate checks to $\tilde O(\sqrt{M})$ checks under coherent access assumptions.
Moreover, BBBV-style oracle lower bounds \cite{BBBV97} show that this quadratic improvement is essentially optimal for unstructured candidate validation.
This delineates a realistic frontier: quantum computation may plausibly accelerate the \emph{search among candidates}, but it cannot compress arbitrary datasets into logarithmically many qubits while retaining worst-case ANN power.

\section{Preliminaries}

\begin{definition}[$c$-Approximate Nearest Neighbor]
Let $(\cX,\dist)$ be a metric space and let $P\subseteq \cX$ be a finite dataset.
For a query $q\in \cX$, write $\dist(q,P)=\min_{p\in P}\dist(q,p)$.
Fix $c\ge 1$.
A (possibly randomized) data structure solves $c$-ANN if for every dataset $P$ and query $q$ it outputs a point $p\in P$ satisfying
\[
\dist(q,p)\le c\cdot \dist(q,P).
\]
\end{definition}

Our lower bound is stated in Hamming space $\bits^d$ with Hamming distance $\Ham(\cdot,\cdot)$.

\begin{definition}
Fix $n$.
A quantum sketch consists of an encoder and a decoder.
The encoder maps each dataset $P=\{p_1,\dots,p_n\}$ to an $m$-qubit density matrix $\rho_P$.
The decoder takes as input a query $q$, receives a copy of $\rho_P$, performs an arbitrary quantum measurement that may depend on $q$ (and may use internal randomness and ancillas), and outputs an index in $[n]$ interpreted as the returned dataset point.
\end{definition}

\begin{remark}
A measurement generally disturbs a quantum state.
All lower bounds below hold in the stronger model where the decoder is given an independent fresh copy of $\rho_P$ for each query.
Therefore the bounds apply a fortiori to any single-copy reusable scheme.
\end{remark}

This model is intentionally permissive.
Any scheme based on amplitude encoding, query-dependent hashing measurements, or other structured designs is a special case.

\begin{definition}[Quantum random access code]
An $(n,m,p)$-QRAC is a map $x\in\bits^n \mapsto \rho_x$ where $\rho_x$ is an $m$-qubit state such that for every index $i\in[n]$ there exists a measurement (depending on $i$) whose outcome equals $x_i$ with probability at least $p$, for every $x$.
\end{definition}

Write $h(p)=-p\log_2 p-(1-p)\log_2(1-p)$ for the binary entropy function.

\begin{theorem}[Lower bound for QRACs \cite{Nay99}]
\label{thm:nayak}
If an $(n,m,p)$-QRAC exists with $p\ge 1/2$, then
\[
m \;\ge\; (1-h(p))\,n.
\]
In particular, for any constant $p>1/2$, one has $m=\Omega(n)$.
\end{theorem}

For completeness, we include a short proof in \cref{sec:proof-nayak}.
The proof rests on Holevo-type information bounds; see also \cite{NC00}.

\section{Lower bound for worst-case sketches}

We now build an explicit family of Hamming instances in which approximate nearest neighbor answers force recovery of $n$ independent bits.

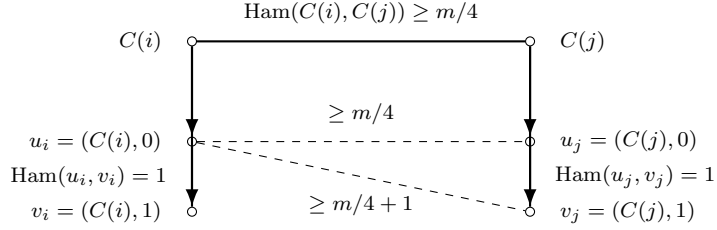
\begin{figure}[t]
\centering
\begin{tikzpicture}[
  font=\small,
  >=Latex,
  node distance=7mm and 10mm,
  blk/.style={draw, rounded corners, inner sep=4pt, align=center},
  pt/.style={draw, circle, inner sep=1.2pt},
  lab/.style={font=\scriptsize},
  arrow/.style={->, thick},
  brace/.style={decorate, decoration={brace, amplitude=4pt}},
]

% Two representative codewords
\node[pt, below left=10mm and 22mm] (Ci) {};
\node[lab, left=2mm of Ci] {$C(i)$};

\node[pt, below right=10mm and 22mm ] (Cj) {};
\node[lab, right=2mm of Cj] {$C(j)$};

% Distance annotation between codewords
\draw[thick] (Ci) -- (Cj);
\node[lab, above=1mm of $(Ci)!0.5!(Cj)$] {$\Ham(C(i),C(j))\ge m/4$};

% u_i, v_i under C(i)
\node[pt, below=12mm of Ci] (ui) {};
\node[lab, left=2mm of ui] {$u_i=(C(i),0)$};

\node[pt, below=8mm of ui] (vi) {};
\node[lab, left=2mm of vi] {$v_i=(C(i),1)$};

\draw[arrow] (Ci) -- (ui);
\draw[arrow] (Ci) -- (vi);

% u_j, v_j under C(j)
\node[pt, below=12mm of Cj] (uj) {};
\node[lab, right=2mm of uj] {$u_j=(C(j),0)$};

\node[pt, below=8mm of uj] (vj) {};
\node[lab, right=2mm of vj] {$v_j=(C(j),1)$};

\draw[arrow] (Cj) -- (uj);
\draw[arrow] (Cj) -- (vj);

% Distance within a pair: 1
\draw[thick] (ui) -- (vi);
\node[lab, left=2mm of $(ui)!0.5!(vi)$] {$\Ham(u_i,v_i)=1$};

\draw[thick] (uj) -- (vj);
\node[lab, right=2mm of $(uj)!0.5!(vj)$] {$\Ham(u_j,v_j)=1$};

% Distances across different codewords remain large (up to +1)
\draw[dashed] (ui) -- (uj);
\node[lab, above=1mm of $(ui)!0.5!(uj)$] {$\ge m/4$};

\draw[dashed] (ui) -- (vj);
\node[lab, below=1mm of $(ui)!0.5!(vj)$] {$\ge m/4+1$};

\end{tikzpicture}
\caption{Intuition for the lower-bound construction. Pick $C(1),\dots,C(n)\in\{0,1\}^m$ so that for all $i\neq j$,
$\Ham(C(i),C(j))\ge m/4$. Then lift each $C(i)$ into a tight pair $u_i=(C(i),0)$ and $v_i=(C(i),1)$ in dimension $d=m{+}1$. The last coordinate creates a distance-$1$ toggle that encodes a bit, while the underlying code separation ensures that all points from different indices remain far apart, enabling the forcing lemma used in the reduction to QRACs.}
\end{figure}

We begin with a standard probabilistic construction of a code on $n$ symbols with block length $\Theta(\log n)$ and constant relative distance.

\begin{lemma}
\label{lem:code}
Let $m \ge 64\log_2 n$.
Then there exist codewords $C(1),\dots,C(n)\in \bits^m$ such that for all $i\neq j$,
\[
\Ham(C(i),C(j)) \;\ge\; \frac{m}{4}.
\]
\end{lemma}

\begin{proof}
Choose $C(1),\dots,C(n)$ independently and uniformly from $\bits^m$.
Fix $i\neq j$.
Then $\Ham(C(i),C(j))$ has the distribution $\mathrm{Bin}(m,1/2)$, with mean $m/2$.
By a Chernoff bound,
\[
\Pr\big[\Ham(C(i),C(j)) \le m/4\big]
=\Pr\Big[\mathrm{Bin}(m,1/2)\le (1-1/2)\cdot (m/2)\Big]
\le \exp(-m/16).
\]
By a union bound over at most $\binom{n}{2} < n^2/2$ pairs,
\[
\Pr\big[\exists\, i\neq j: \Ham(C(i),C(j)) \le m/4\big]
\le \frac{n^2}{2}\, e^{-m/16}.
\]
If $m\ge 64\log_2 n$, then $n^2 e^{-m/16} \le n^2 e^{-4\ln n} = n^{-2} < 1$, hence the above probability is strictly less than $1$, implying the existence of a choice of codewords with pairwise distance at least $m/4$.
\end{proof}

Fix any approximation factor $c\ge 1$ and success probability $p>1/2$.
Let $m$ be large enough so that $m\ge 64\log_2 n$ and $m/4 \ge c+1$.
Set $d=m+1$.
Let $C(1),\dots,C(n)\in\bits^m$ be codewords guaranteed by \cref{lem:code}.
Define, for each $i\in[n]$,
\[
u_i = C(i)\circ 0 \in \bits^d,
\qquad
v_i = C(i)\circ 1 \in \bits^d,
\]
where $\circ$ denotes concatenation and the last coordinate is the appended bit.

For each bit string $x\in\bits^n$, define a dataset $P_x=\{p_1,\dots,p_n\}\subseteq \bits^d$ by
\[
p_i \;=\;
\begin{cases}
u_i & \text{if } x_i=0,\\
v_i & \text{if } x_i=1.
\end{cases}
\]
Define the query points $q_i=u_i$ for all $i\in[n]$.

\begin{lemma}
\label{lem:forcing}
For every $x\in\bits^n$ and every $i\in[n]$, any correct $c$-ANN answer to query $q_i$ on dataset $P_x$ must be $u_i$ if $x_i=0$ and must be $v_i$ if $x_i=1$.
\end{lemma}

\begin{proof}
If $x_i=0$ then $u_i\in P_x$ and $\Ham(q_i,u_i)=\Ham(u_i,u_i)=0$, hence $\dist(q_i,P_x)=0$.
The $c$-ANN condition requires output distance at most $c\cdot 0=0$, so the output must be a dataset point at Hamming distance $0$ from $q_i$, which is uniquely $u_i$.

If $x_i=1$ then $v_i\in P_x$ and $q_i=u_i$ differs from $v_i$ only in the last coordinate, so $\Ham(q_i,v_i)=1$ and therefore $\dist(q_i,P_x)\le 1$.
We now show that every other dataset point is far from $q_i$.
Fix $j\neq i$.
If $x_j=0$, then $p_j=u_j=C(j)\circ 0$ and
\[
\Ham(q_i,p_j)=\Ham(C(i)\circ 0,\, C(j)\circ 0)=\Ham(C(i),C(j))\ge m/4 \ge c+1.
\]
If $x_j=1$, then $p_j=v_j=C(j)\circ 1$ and
\[
\Ham(q_i,p_j)=\Ham(C(i)\circ 0,\, C(j)\circ 1)=\Ham(C(i),C(j))+1 \ge m/4+1 \ge c+2.
\]
Thus every $p_j$ with $j\neq i$ is at distance at least $c+1$ from $q_i$, while $v_i$ is at distance $1$.
Since $\dist(q_i,P_x)\le 1$, any $c$-approximate neighbor must be within distance $c\cdot 1=c$ of $q_i$, and therefore cannot be any $p_j$ for $j\neq i$.
The unique valid output is $v_i$.
\end{proof}

\begin{theorem}
\label{thm:main}
Fix any $c\ge 1$ and any success probability $p>1/2$.
For the dataset family $\{P_x\}_{x\in\bits^n}\subseteq \bits^d$ with $d=\Theta(\log n)$ constructed above and the queries $\{q_i\}_{i=1}^n$, any quantum sketch that encodes $P_x$ into an $m$-qubit state $\rho_{P_x}$ and answers $c$-ANN queries $q_i$ with success probability at least $p$ for every $x$ and $i$ must satisfy
\[
m \;\ge\; (1-h(p))\,n \;=\; \Omega(n).
\]
In particular, $m=O(\log n)$ qubits is impossible.
\end{theorem}

\begin{proof}
Assume a sketch exists.
Define an encoding of $x\in\bits^n$ by $\rho_x := \rho_{P_x}$.
Given an index $i\in[n]$, run the ANN query procedure on input $q_i$ using a fresh copy of $\rho_x$.
By \cref{lem:forcing}, any correct answer must be $u_i$ if $x_i=0$ and $v_i$ if $x_i=1$, hence the returned index determines $x_i$.
By the assumed success probability, this recovers $x_i$ with probability at least $p$ for every $x$ and $i$.
Therefore $\{\rho_x\}$ is an $(n,m,p)$-QRAC.
By Nayak's lower bound (\cref{thm:nayak}), $m\ge (1-h(p))n$.
\end{proof}

\begin{remark}
The obstruction is not that the dimension is too small: the construction already lives in dimension $d=\Theta(\log n)$, the same scale suggested by JL-type reductions.
What matters is that the admissible dataset family contains $2^n$ instances whose query-answer behavior separates all $n$ bits.
Geometric compression of coordinates does not prevent the query mechanism from being forced to reveal independent dataset information.
\end{remark}

A hashing-based ANN scheme can be viewed abstractly as preprocessing the dataset into a short description (hash tables), and at query time probing a few buckets and returning an element.
In our sketch model, the entire preprocessing output is absorbed into $\rho_P$, and the query procedure is an arbitrary measurement that may implement any internal hashing, bucket selection, and postprocessing.
Thus \cref{thm:main} rules out \emph{all} schemes---classical or quantum---that attempt to compress an arbitrary $n$-point dataset into $O(\log n)$ qubits while preserving worst-case ANN power.

\section{Capacity viewpoint}

The explicit family $\{P_x\}$ is a concrete vehicle for a more general lesson: the relevant parameter is the combinatorial richness of the query-answer behavior induced by the allowed dataset family.

Consider a decision version of near-neighbor search.
Fix a threshold $r>0$ and define
\[
\mathrm{Near}_r(P,q) \;:=\; \mathbf{1}\big[\exists p\in P:\dist(p,q)\le r\big].
\]
Each dataset $P$ induces a Boolean function $f_P(\cdot)=\mathrm{Near}_r(P,\cdot)$ on a query domain $\cQ$.
Let $\cF=\{f_P: P\in\cP\}$ be the induced function class from an allowed dataset family $\cP$.

\begin{proposition}
\label{prop:vc}
If $\cF$ shatters a set of $t$ queries (i.e., $\mathrm{VCdim}(\cF)\ge t$), then any $m$-qubit quantum sketch that answers $\mathrm{Near}_r(P,q)$ correctly on those $t$ queries with success probability at least $p>1/2$ must satisfy
\[
m \;\ge\; (1-h(p))\,t.
\]
\end{proposition}

\begin{proof}
Shattering means there exist queries $q_1,\dots,q_t$ such that for every label string $x\in\bits^t$ there exists a dataset $P_x\in \cP$ with $f_{P_x}(q_i)=x_i$ for all $i$.
Given a sketch for the decision problem, define $\rho_x:=\rho_{P_x}$.
To decode bit $x_i$, run the decision procedure on $q_i$ using a fresh copy of $\rho_x$ and output the resulting bit.
By correctness, this succeeds with probability at least $p$, hence $\{\rho_x\}$ forms a $(t,m,p)$-QRAC.
Apply \cref{thm:nayak}.
\end{proof}

\begin{remark}
Nearest neighbor search outputs an index (a multi-class label), not a bit.
In learning theory, the Natarajan dimension generalizes VC dimension to multi-class hypothesis classes \cite{Nat89}.
One can either reduce ANN to a decision problem as above or work directly with multi-class shattering.
Either route leads to the same message: large answer-function capacity forces large memory, even when the memory is quantum.
\end{remark}

\section{What quantum improvements remain}

The lower bound concerns compressing the dataset into a single short quantum state.
It does not rule out quantum advantages when the dataset is stored in a more conventional memory model and quantum computation accelerates the \emph{query phase}.
A common pattern in hashing-based ANN is that hashing produces a candidate set of indices, and one then checks candidates by evaluating distances or a predicate.

We formalize this as a candidate-scanning oracle model.

\begin{definition}
\label{def:candidate}
Fix a query $q$ and a candidate set $S(q)\subseteq [n]$ with $|S(q)|=M$.
Assume coherent access in the following sense.
There is a procedure that prepares the uniform superposition
\[
\frac{1}{\sqrt{M}}\sum_{j\in S(q)} |j\rangle,
\]
and there is a phase oracle $O$ that, given $|j\rangle$, flips the phase iff $j\in S(q)$ satisfies a desired predicate (e.g.\ $\dist(q,p_j)\le r$).
Let $t$ be the number of satisfying indices in $S(q)$.
\end{definition}

\begin{theorem}[\cite{Gro96}]
\label{thm:grover}
In the model of \cref{def:candidate}, if $t\ge 1$ then there is a quantum algorithm that finds a satisfying index with constant probability using $O(\sqrt{M/t})$ oracle calls.
\end{theorem}

\begin{proof}
Let $|G\rangle$ be the normalized uniform superposition over satisfying indices and $|B\rangle$ be the normalized uniform superposition over non-satisfying indices.
The initial state is
\[
|\psi_0\rangle \;=\; \sqrt{\frac{t}{M}}\,|G\rangle \;+\; \sqrt{1-\frac{t}{M}}\,|B\rangle,
\]
which lies in the two-dimensional span of $\{|G\rangle,|B\rangle\}$.
A Grover iterate consists of a reflection about the marked subspace (implemented by $O$) followed by a reflection about $|\psi_0\rangle$.
This acts as a rotation by angle $2\theta$ in the $\{|G\rangle,|B\rangle\}$ plane, where $\sin^2\theta=t/M$.
After $k$ iterations, the amplitude on $|G\rangle$ becomes $\sin((2k+1)\theta)$, which is a constant for $k=\Theta(1/\theta)=\Theta(\sqrt{M/t})$.
Measuring then yields a satisfying index with constant probability.
\end{proof}

\begin{theorem}[Black-box optimality \cite{BBBV97}]
\label{thm:bbbv}
In the candidate-scanning oracle model, any quantum algorithm that finds a satisfying index with constant probability requires $\Omega(\sqrt{M/t})$ oracle calls in the worst case.
In particular, when $t=1$ one needs $\Omega(\sqrt{M})$ calls.
\end{theorem}

\begin{proof}[Proof sketch]
When $t=1$, the oracle restricted to the candidate register is precisely the unstructured search oracle with a unique marked item.
The BBBV hybrid argument bounds how much $Q$ oracle calls can separate the final states corresponding to different marked positions.
Concretely, averaging over the marked index $s\in[M]$, the expected squared distance between the final state under oracle $O_s$ and the final state under the all-zero oracle is $O(Q^2/M)$.
If $Q=o(\sqrt{M})$, these states remain close for most $s$, hence no measurement can identify $s$ with constant success probability.
The general $t$ case follows by standard reductions (e.g.\ by embedding a size-$M/t$ unique-search instance into the promise of $t$ marked items).
\end{proof}

\subsection{Relation to oracle-model quantum nearest neighbor algorithms}

Several quantum nearest neighbor proposals operate in oracle/QRAM-like models where training vectors are accessed coherently and distances are estimated in superposition.
For instance, Wiebe, Kapoor, and Svore \cite{WKS14} develop quantum algorithms for nearest-neighbor learning that combine coherent distance estimation with quantum minimum finding and achieve query-complexity scalings consistent with $\tilde O(\sqrt{M})$-type improvements in regimes where the essential difficulty reduces to searching among candidates.
This aligns with \cref{thm:grover,thm:bbbv}.
Our memory lower bound is compatible with such results because it targets a different goal: compressing \emph{arbitrary} datasets into a single $O(\log n)$-qubit state while preserving worst-case query power.

\section{Discussion}

The obstruction proved here is conceptually simple.
Worst-case ANN can be forced to reveal $n$ essentially independent bits of dataset information through $n$ carefully chosen queries, and no $m=o(n)$ qubit state can support decoding of all those bits with constant bias.
The fact that the hard family lives already in dimension $d=\Theta(\log n)$ underscores that JL-type reductions do not address the bottleneck; the bottleneck is not coordinate representation size but recoverable classical information.

At the same time, the lower bound is not a ``no quantum advantage'' statement.
It suggests that the most robust quantum gains in ANN-like pipelines occur when the dataset is stored in a standard memory model but the query algorithm obtains coherent access to candidate sets and predicates.
In that setting, quadratic speedups in candidate scanning are both achievable and essentially optimal absent additional structure.

A natural direction is to identify restricted dataset families that avoid the random access code obstruction.
The capacity viewpoint (\cref{prop:vc}) suggests one route, if the induced answer functions have small combinatorial dimension on the query domain of interest, then information-theoretic compression may be possible.
The harder question is which restrictions remain meaningful for practical similarity search and which admit algorithmic sketches that are simultaneously small and efficiently decodable.

Another direction is to move beyond black-box candidate checking.
If bucket contents have additional structure---for example, strong promise gaps, special algebraic structure, or geometric regularity that can be exploited coherently---then the right quantum primitives need not be plain Grover search, and stronger-than-quadratic speedups may become possible for carefully promised instances.

% \section{Acknowledgments}
% The author declares that there are no competing interests, financial or otherwise, that could have influenced this work. This research did not receive any specific grant from funding agencies in the public, commercial, or not-for-profit sectors.
% \subsection{Declaration of generative AI and AI-assisted technologies in the manuscript preparation process}
% During the preparation of this work, the author used ChatGPT, a generative AI language model, in order to assist with drafting portions of the text and checking for editorial and linguistic errors. After using this tool, the author reviewed and edited the content as needed and takes full responsibility for the content of the published article.

\bibliographystyle{alpha}
\bibliography{references}

\appendix

\section{Proof of Nayak's QRAC lower bound}
\label{sec:proof-nayak}

We provide a self-contained proof of \cref{thm:nayak} following a standard information-theoretic route \cite{Nay99,NC00}.

Let $X=(X_1,\dots,X_n)$ be uniformly random in $\bits^n$ and consider the classical--quantum state
\[
\omega_{XB} \;=\; 2^{-n}\sum_{x\in\bits^n} |x\rangle\!\langle x| \otimes \rho_x,
\]
where $B$ denotes the $m$-qubit system.
Let $\rho := 2^{-n}\sum_x \rho_x$ be the average state on $B$.

We first upper bound the mutual information $I(X:B)_\omega$.
Since $X$ is classical, one has
\[
I(X:B)_\omega \;=\; S(\rho) - 2^{-n}\sum_{x} S(\rho_x),
\]
where $S(\cdot)$ denotes von Neumann entropy.
In particular $I(X:B)_\omega \le S(\rho)$.
Because $\rho$ acts on an $m$-qubit system, its entropy is at most $m$ bits, i.e.\ $S(\rho)\le m$.
Thus
\begin{equation}
\label{eq:holevo}
I(X:B)_\omega \;\le\; m.
\end{equation}

Next we lower bound $I(X:B)_\omega$ using the QRAC decoding guarantee.
By the chain rule for mutual information,
\[
I(X:B)_\omega \;=\; \sum_{i=1}^n I(X_i : B \mid X_{<i})_\omega,
\]
where $X_{<i}=(X_1,\dots,X_{i-1})$.

Fix $i$ and condition on any value $a\in\bits^{i-1}$ of $X_{<i}$.
Because the QRAC decoding for index $i$ succeeds with probability at least $p$ for every input string $x$, it in particular succeeds with probability at least $p$ under the conditional distribution $X\mid (X_{<i}=a)$.
Let $M_i$ be the measurement used to decode $X_i$.
Apply $M_i$ to the quantum system $B$ and let $G_i$ be the resulting classical guess bit.
Then, conditioned on $X_{<i}=a$, we have
\[
\Pr[G_i = X_i \mid X_{<i}=a] \;\ge\; p.
\]
Since $X_i$ is uniform in $\bits$ even after conditioning on $X_{<i}=a$, the pair $(X_i,G_i)$ is a binary symmetric channel with error probability at most $1-p$, hence
\[
H(X_i \mid G_i, X_{<i}=a)\;\le\; h(p).
\]
Therefore
\[
I(X_i : G_i \mid X_{<i}=a)
= H(X_i\mid X_{<i}=a) - H(X_i\mid G_i, X_{<i}=a)
\ge 1 - h(p),
\]
because $H(X_i\mid X_{<i}=a)=1$.
Finally, since $G_i$ is obtained from $B$ by a measurement, the data-processing inequality implies
\[
I(X_i : B \mid X_{<i}=a) \;\ge\; I(X_i : G_i \mid X_{<i}=a) \;\ge\; 1-h(p).
\]
Averaging over $a$ yields $I(X_i:B\mid X_{<i})_\omega \ge 1-h(p)$, and summing over $i$ gives
\begin{equation}
\label{eq:lower}
I(X:B)_\omega \;\ge\; n(1-h(p)).
\end{equation}

Combining \eqref{eq:holevo} and \eqref{eq:lower} yields $m \ge n(1-h(p))$, which is \cref{thm:nayak}.
\qed
\end{document}